%% file: gandalf16.tex
\documentclass[submission,copyright,creativecommons]{eptcs}
\providecommand{\event}{GandALF 2016} 

\usepackage{graphicx}
\usepackage{amsmath}
\usepackage{amssymb}
\usepackage{amsthm}
\usepackage{tikz}
\usepackage{xspace}

\usepackage{color}

\include{macros}

\definecolor{hypblue}{rgb}{0,0,1}
\definecolor{hypred}{rgb}{1,0,0}
\hypersetup{colorlinks,
linkcolor=hypred,
citecolor=hypblue
}

\newtheorem{theorem}{Theorem}
\newtheorem{definition}{Definition}

\newtheorem{corollary}{Corollary}

\title{Relation-Changing Logics as Fragments of Hybrid Logics}

\author{Carlos Areces \institute{Universidad Nacional de
    C\'ordoba,\\ Argentina} \institute{CONICET, Argentina}
\and 
Raul Fervari \institute{Universidad Nacional de
    C\'ordoba,\\ Argentina} \institute{CONICET, Argentina}
\and 
Guillaume Hoffmann\institute{Universidad Nacional de
    C\'ordoba,\\ Argentina} \institute{CONICET, Argentina}
\and Mauricio Martel\institute{Universit\"at Bremen, Germany}}

\def\titlerunning{Relation-Changing Logics as Fragments of Hybrid Logics}
\def\authorrunning{Areces, Fervari, Hoffmann and Martel}

\begin{document}
\maketitle

  \begin{abstract}
    Relation-changing modal logics are extensions of the basic modal
    logic that allow changes to the accessibility relation of a model
    during the evaluation of a formula. In particular, they are
    equipped with dynamic modalities that are able to delete, add, and
    swap edges in the model, both locally and globally.  We provide
    translations from these logics to hybrid logic
    along with an implementation.
    In general, these logics are undecidable, but we use our translations
    to identify decidable fragments.
    We also compare the expressive
    power of relation-changing modal logics with hybrid logics.
  \end{abstract}

\input{intro}
\input{basic}
\input{hybrid}
\input{translations}
\input{implementation}
\input{decidable}

\input{exppow}

\input{final}

\paragraph{Acknowledgments:} This work was partially supported by
grant ANPCyT-PICT-2013-2011, DFG grant LU 1417/2-1,  
and the Laboratoire International Associ\'e ``INFINIS''.

\bibliographystyle{eptcs}
\bibliography{gandalf16}

\end{document}

%% file: macros.tex
\newcommand\dswap{\tup{{\sf sw}}}
\newcommand\bswap{[{\sf sw}]}

\newcommand\dgswap{\tup{{\sf gsw}}}

\newcommand{\dsab}{\tup{{\sf sb}}}
\newcommand{\bsab}{[{\sf sb}]}

\newcommand\dgsab{\tup{{\sf gsb}}}

\newcommand{\dbridge}{{\tup{\sf br}}}

\newcommand{\dgbridge}{{\tup{\sf gbr}}}

\newcommand{\down}{{\downarrow}}

\newcommand\sab{\text{%
    \setbox0\hbox{$\Diamond$}%
    \rlap{\hbox to \wd0{\hss\hss\hss\raisebox{.05\height}{\hspace*{-1pt}$\textendash$}\hss}}\box0
}}

\newcommand*{\greysquare}{\textcolor{gray}{\blacksquare}}

\newcommand{\bml}{\ensuremath{\mathcal{ML}}\xspace}

\newcommand{\Hl}{\mathcal{HL}}

\newcommand{\fol}{\mathcal{FOL}}

\renewcommand{\L}{{\cal L}}

\newcommand{\PROP}{{\rm \sf PROP}\xspace}

\newcommand{\NOM}{{\rm \sf NOM}\xspace}
\newcommand{\FORM}{{\rm \sf FORM}\xspace}

\newcommand{\model}{\mathcal M}

\newcommand{\E}{{\sf E}}
\newcommand{\A}{{\sf A}}

\newcommand{\tup}[1]{\langle #1 \rangle}

\newcommand{\cset}[1]{\{ #1 \}}

\newcommand{\pow}[1]{\mathcal P #1}

\renewcommand{\iff}{\mbox{iff}}
\newcommand{\ra}{\rightarrow}

\newcommand{\Tr}{{\sf Tr}}

\usepgflibrary{arrows}
\usetikzlibrary{arrows,automata,shapes,decorations.markings,decorations.pathmorphing,
backgrounds,fit,calc}                                                   

\tikzset{
    >=stealth',
    ar/.style={
           shorten <=2pt,
           shorten >=2pt,}
}

\tikzstyle{world} = [shape=circle,fill=white,inner sep=2pt,draw=black,blur shadow={shadow xshift=0ex,shadow yshift=0ex,shadow scale=1.1}]
\tikzstyle{blackworld} = [shape=circle,fill=black,inner sep=2pt,draw=black,blur shadow={shadow xshift=0ex,shadow yshift=0ex,shadow scale=1.1}]

\tikzstyle{vecArrow} = [thick, decoration={markings,mark=at position
   1 with {\arrow[semithick]{open triangle 60}}},
   double distance=1.4pt, shorten >= 5.5pt,
   preaction = {decorate},
   postaction = {draw,line width=1.4pt, white,shorten >= 4.5pt}]

\usepackage[textwidth=.75in,textsize=scriptsize]{todonotes}

%% file: intro.tex
\section{Introduction}
\label{sec:intro}

Modal logics~\cite{blackburn06,blackburn01} were originally conceived
as logics of necessary and possible truths.  They are now viewed, more
broadly, as logics that explore a wide range of modalities, or modes
of truth: epistemic (``it is known that''), doxastic (``it is believed
that''), deontic (``it ought to be the case that''), or temporal (``it
has been the case that''), among others. From a model-theoretic
perspective, the field evolved into a discipline that deals with
languages interpreted on various kinds of relational structures or
graphs.  Nowadays, modal logics are actively used in areas as diverse
as software verification, artificial intelligence, semantics and
pragmatics of natural language, law, philosophy, etc. 

As we just mentioned, from an abstract point of view, modal logics can
be seen as formal languages to navigate and explore properties of a
given relational structure.  If we are interested, on the other hand,
in describing how a given relational structure evolves (through time
or through the application of certain operations) then classical modal
languages seem a priori to fall short of the mark. Of course, it is
always possible to model the whole space of possible transformations
as a graph, and use modal languages at that level, but this soon
becomes unwieldy.  It would be more elegant to use \emph{dynamic
modalities} that mimic the changes the structure will undergo.

There exist several dynamic modal operators that fit in this
approach.  A clear example are the dynamic operators introduced in dynamic 
epistemic logics (see, e.g.,~\cite{vanditmarsch07}). These operators are used to model 
changes in the epistemic state of an agent by removing edges
from the graph that represents the information states the agent 
considers possible.  A less obvious example is given by hybrid 
logics~\cite{arec:hybr05b,blackburn95}
equipped with the down arrow operator $\down$ which is used to 
`rebind' names for states to the current point of evaluation. 
Finally, a classical example is Sabotage Logic introduced 
by van Benthem in~\cite{vanbenthem05}. The sabotage operator deletes individual edges in a graph 
and was introduced to solve the \emph{sabotage game}. 
This game is played on a graph by two players, {\em Runner} and {\em Blocker}.
Runner can move on the graph from node to accessible node, starting from a designated point, 
and with the goal of reaching a given final point.
Blocker, on the other hand, can delete one edge from the graph every time 
it is his turn. Runner wins if he manages to move from the origin to the final point, while Blocker wins otherwise.  
Van Benthem turns the sabotage game 
into a modal logic, where the (global) sabotage operator  $\dgsab$ models the moves of Blocker,
and is interpreted on a graph $\model$ at a point $w$ as:
$$
\model,w \models \dgsab \varphi ~ \mbox{ iff } ~ \mbox{ there is a pair $(u, v)$ of $\model$ such that $\model^-_{(u,v)}, w  \models \varphi$}
$$
where $\model^-_{(u,v)}$ is identical to $\model$ except that the edge
$(u,v)$ has been removed.  The moves of Runner, on the other hand, can 
be modeled using the standard $\Diamond$ operator of classical modal logics. 

More recently, Sabotage Logic was proposed as a formalism for reasoning about formal learning theory~\cite{GKVQ09}.
Learning can be seen as a game with two players, {\em Teacher} and {\em Learner},  
where Learner changes his information state through a step-by-step process.
The process is successful if he eventually reaches an information 
state describing the real state of affairs.
The information that Teacher provides 
can be interpreted as feedback about Learner's conjectures about the current state of affairs, allowing him to discard inconsistent
hypotheses.
It should be clear that from this game-theoretical perspective, the interaction between Teacher and Learner 
can be modeled using Sabotage Logic. 

The dynamic approach seems appealing and very flexible: it is easy to come up with situations that nicely fit and 
extend the examples we just mentioned. Discovering alternative routes for Runner in van Benthem's sabotage game, or possible shortcuts that 
Learner can take in learning theory can be modeled by adding new edges to the graph. 
Swapping an edge can be used to represent other scenarios such as changing the direction of a route, 
or allowing Learner to return to a previous information state. All these primitives can also 
be turned into a modal logic in the same way as Sabotage Logic,
in order to get a formal language for reasoning about the games. 

Motivated by scenarios like the ones we just described, we investigate three dynamic primitives that can 
change the accessibility relation 
of a model:  {\em sabotage} (deletes edges from the model), {\em bridge} (adds edges to a model), and {\em swap} 
(turns around edges), both in a global version (performing changes anywhere in the model) and local 
(changing adjacent edges from the evaluation point).
We have chosen these relation-changing operators with the intention of covering a sufficiently varied sample of alternatives, as it is done in previous works. In~\cite{areces12} we first introduced the primitives, and we investigated their expressive
power and model checking problem.  We introduced tableaux methods for relation-changing modal logics in~\cite{areces13}.
In~\cite{areces14igpl} we studied local swap logic, in particular its decidability problem and its relation with first-order and hybrid logics.
In~\cite{AFH15} a general framework for representing
model updates is defined, and connections with dynamic epistemic logic were introduced in~\cite{areces14wollic,ADFS15}.
Finally, we know that the satisfiability problem for the six relation-changing logics considered
is undecidable~\cite{loding03fsttcs,fervari14phd,areces14igpl,martel15}. 

In this article, we show that relation-changing logics can be seen as fragments of hybrid logics.
We consider hybrid logics because it is the best known \emph{modal} logic that can
simulate the semantics of relation-changing operators.
We introduce translations 
to $\Hl(\E,\down)$, the basic modal logic
extended with nominals, the down arrow binder $\down$, and the universal modality $\E$ (in some cases the translations fall into
the less expressive hybrid logic $\Hl(:,\down)$, i.e., with the satisfiability operator $:$ instead
of $\E$).
We also show that relation-changing logics are strictly less expressive than the hybrid logics they are translated into. Then, we discuss how
we can benefit from known decidable fragments of
$\Hl(\E,\down)$ to find decidable fragments of relation-changing modal logics. Finally, we have implemented these 
translations into the hybrid logic prover HTab~\cite{Hoffmann2007} so that it can now reason on relation-changing formulas. 

The article is organized as follows. In Section~\ref{sec:basic} we introduce the syntax and semantics
of relation-changing modal logics. Section~\ref{sec:hybrid} introduces different hybrid
extensions of modal logic which are used in Section~\ref{sec:trans} to encode relation-changing operators.
An implementation is described in Section~\ref{sec:imple} with some examples.
Then, in Section~\ref{sec:decid} we discuss decidability results, and 
in Section~\ref{sec:exppow} we compare the expressivity of relation-changing logics and hybrid logics.
Finally, we conclude with some remarks and
future work in Section~\ref{sec:final}.


%% file: basic.tex
\section{Relation-Changing Modal Logics}
\label{sec:basic}

In this section, we formally introduce extensions of the basic
modal logic with relation-changing operators. For more details,
we direct the reader to, e.g.,~\cite{fervari14phd}.

\begin{definition}[Syntax]\label{def:relchangsyn}
Let $\PROP$ be a countable, infinite set of propositional symbols.
The set $\FORM$ of formulas over $\PROP$ is defined as:
$$
\FORM ::= \bot \mid p \mid  \neg \varphi
  \mid \varphi \land \psi 
  \mid \lozenge \varphi
  \mid \blacklozenge \varphi,
$$
where $p\in \PROP$, $\blacklozenge \in \{\dsab, \dbridge, \dswap, \dgsab, \dgbridge,\dgswap\}$,
and $\varphi, \psi \in \FORM$. Other operators are defined
as usual. 

Let $\bml$ (the basic modal logic) be the logic without the $\cset{\dsab, \dbridge, \dswap, \dgsab, \dgbridge,\dgswap}$ operators,
and $\bml(\blacklozenge)$ the extension of $\bml$ allowing also
$\blacklozenge$, for $\blacklozenge \in \cset{\dsab, \dbridge, \dswap,
  \dgsab, \dgbridge,\dgswap}$. In particular, $\bml(\dsab, \dgsab)$, $\bml(\dbridge, \dgbridge)$, and $\bml(\dswap, \dgswap)$ will be called Sabotage Logic, Bridge Logic, and Swap Logic, respectively.
\end{definition}

Semantically, formulas are evaluated in standard relational models, and the meaning of the operators of the basic modal logic remains unchanged (see~\cite{blackburn01} for details). When we evaluate formulas containing relation-changing operators, we will need to keep track of the edges that have been modified. To that end, let us define precisely the models that we will use. 

\begin{definition}[Models and model updates]\label{def:rcmodels} 
A model $\model$ is a triple $\model = \langle W, R, V \rangle$,
where $W$ is a non-empty set whose elements are called points or states;
$R \subseteq W{\times} W$ is the accessibility relation; and $V: \PROP \ra \pow(W)$ 
is a valuation. We define the following notations:
\begin{center}
\begin{tabular}{ll}
\textbf{(sabotaging)} & $\model^-_{S} = \langle W, R^-_{S}, V \rangle$, 
with $R^-_{S} = R{\setminus} S$, $ S\subseteq R$.\\
\textbf{(bridging)} &  $\model^+_{S} =  \langle W, R^+_{S}, V \rangle$, 
with $R^+_{S}  = R \cup S$, $S \subseteq (W{\times}W){\setminus}R$.\\
\textbf{(swapping)} & $\model^*_{S} = \langle W, R^*_{S}, V \rangle$, 
with $R^*_{S} = (R{\setminus} S^{-1} ){\cup}S$,  $S \subseteq R^{-1}$.
\end{tabular}
\end{center}
\end{definition}

Intuitively, $\model^-_{S}$ is obtained from $\model$ by deleting the edges in $S$, and similarly 
$\model^+_S$ adds the edges in $S$ to the accessibility relation, and $\model^*_S$ adds the 
edges in $S$ as inverses of edges previously in the accessibility relation. 

Let $w$ be a state in $\model$, the pair $(\model,w)$ is called a pointed model;
we will usually drop parenthesis and write $\model,w$ instead of $(\model,w)$. In the
rest of this article, we will use $wv$ as a shorthand for $\{(w,v)\}$ or $(w,v)$;
context will always disambiguate the intended use.

\begin{definition}[Semantics]\label{def:relchangsem}
Given a pointed model $\model,w$ and a formula $\varphi$, we say that $\model,w$
satisfies $\varphi$, and write $\model,w \models \varphi$, when
$$
\begin{array}{lcl}
\model,  w \models p & \iff & \  w \in V(p) \\
\model,  w \models \neg \varphi & \iff & \ \model,  w \not\models \varphi\\
\model,  w \models \varphi \land \psi & \iff
 & \ \model,  w \models \varphi \mbox{ and } \model,  w \models \psi\\
\model,  w \models \lozenge \varphi & \iff
  & \mbox{ for some }  v \in W \mbox{ s.t. } (w, v) \in R,
        \model,  v  \models \varphi\\
\model,  w \models \dsab \varphi & \iff & \mbox{ for some }  v \in W
  \mbox{ s.t. } (w, v) \in R,   \model^-_{wv},  v  \models \varphi\\
\model,  w \models \dbridge \varphi & \iff & \mbox{ for some }  v \in W
  \mbox{ s.t. } (w, v) \not\in R,   \model^+_{wv},  v  \models \varphi \\
\model,  w \models \dswap \varphi & \iff & \mbox{ for some }  v \in W
  \mbox{ s.t. } (w, v) \in R,   \model^*_{vw},  v  \models \varphi\\
\model,  w \models \dgsab \varphi & \iff & \mbox{ for some }  v,u\in W,
            \mbox{ s.t. } (v,u) \in R, \model^-_{vu},  w  \models \varphi\\
\model,  w \models \dgbridge \varphi & \iff & \mbox{ for some }  v,u\in W,
            \mbox{ s.t. } (v, u) \not\in R, \model^+_{vu},  w  \models \varphi\\
\model,  w \models \dgswap \varphi & \iff & \mbox{ for some }  v,u\in W,
            \mbox{ s.t. } (v, u) \in R, \model^*_{uv},  w  \models \varphi.
\end{array}
$$

We say that $\varphi$ is satisfiable if for some pointed model $\model, w$ we have 
$\model, w \models \varphi$. 
\end{definition}

The meaning of the relation-changing operators $\dsab$ (local sabotage), $\dbridge$ (local bridge), 
$\dswap$ (local swap), $\dgsab$ (global sabotage), $\dgbridge$ (global bridge) and $\dgswap$ (global swap) should
be clear from the semantic conditions above.  The local operators alter one arrow which is adjacent to the point of 
evaluation (deleting, adding and swapping it, respectively) while the global versions can change an arrow 
anywhere in the model.


%% file: hybrid.tex
\section{Extensions of Modal Logic and Hybrid Logic}\label{sec:hybrid}

In this section, we present several extensions of the basic modal logic $\bml$.
The existential modality~\cite{gorausin92}, written $\E \varphi$, extends $\bml$ in the following way:
$$
\begin{array}{lcl}
\model,  w \models \E \varphi & \iff
  & \mbox{ for some }  v \in W, ~ \model,  v  \models \varphi.
\end{array}
$$
In words, $\E\varphi$ is true at a state $w$ if $\varphi$ is true somewhere in the model. 
The $\E$ operator, with its dual $\A$, has been extensively investigated in classical modal logic~\cite{Spaan93}.

Now we consider several traditional `hybrid' operators (see~\cite{arec:hybr05b} for details):
nominals, the satisfaction operator, and the down-arrow binder.
The basic hybrid logic $\Hl$ is obtained by adding \emph{nominals} to $\bml$.
A nominal is a
propositional symbol that is true at exactly one state in a model. Fix the signature $\tup{\PROP, \NOM}$,
with $\NOM \subseteq \PROP$.
For $n \in \NOM$, we require that its valuation is a singleton set,
i.e., there is a single state $w$ such that $V(n) = \{w\}$.
In addition to nominals, hybrid logic typically involves
the \emph{satisfaction operator}. Given a nominal $n$ and a formula
$\varphi$, the satisfaction operator is written $n:\varphi$. The intended meaning is ``$\varphi$ is true at the
state named by $n$''. Its semantics is given by the following clause:
$$
\begin{array}{rcl}
{\cal M}, w \models n:\varphi
   & \mbox{iff} 
   & \model, v \models \varphi \mbox{ where } V(n) = \{v\}.
\end{array}
$$

Observe that if the language has the $\E$ operator and nominals, then
$n:\varphi$ is definable because  $n:\varphi$ is equivalent
to $\E(n \wedge \varphi)$. 

Finally, consider the \emph{down-arrow binder} operator, written $\downarrow$.
Let the valuation $V_n^w$ be defined by $V_n^w(n)= \{w\}$ and $V_n^w(m)= V(m)$,
when $n \not = m$.
The semantic condition for $\down$ is the following:
$$
\begin{array}{rcl}
\tup{W,R,V}, w \models \down n . \varphi 
   & \mbox{iff} 
   & \tup{W,R,V_n^w}, w \models \varphi.
\end{array}
$$

The language $\Hl(:,\down)$ is a reduction class of first-order logic, and
is thus undecidable~\cite{blackburn95,tencate_phd}.
It remains undecidable even with a single accessibility
relation, no satisfaction operator, and only nominal propositional symbols~\cite{arecroad99}.
$\Hl(\E,\down)$ is equivalent to first-order logic,
since $\down$ can define the operators $\exists$ and
$\forall$ when combined with $\E$ and $\A$.

Contrary to relation-changing modal logics, the logic $\Hl(\E, \down)$
is not able to modify the accessibility relation of a model. However,
it can use the binder to name as many states as needed. Being able to
name states implies that we can also name \emph{specific edges} in the model.
This is what the translations introduced in the next section will exploit.


%% file: translations.tex
\section{Translations to Hybrid Logics}
\label{sec:trans}

Relation-changing (RC) logics and hybrid logics with the binder
$\down$ are two families of logics that are dynamic in their own way.
The dynamicity of RC logics is quite obvious: they are able to modify
the accessibility relation in a model in an explicit way.  On the
other hand, hybrid logics carefully move nominals around, avoiding to
touch anything else in the model. If we consider both formalisms, it
would seem that hybrid logics are the gentler and weaker of both.
However, this is not true. Hybrid logics have the advantage of
surgical precision over RC logics. Being able to name states of the
model and use these names turns out to be a crucial advantage. As we
will see now, naming can be used to manipulate \emph{edges} by naming
pairs of states using the pattern $\down x . \Diamond \down y
. \varphi$.  We use this naming technique to simulate edge deletion,
addition, and swapping. 

Our translations are parametrized over a set
of pair of nominals $S \subseteq \NOM \times \NOM$. For a given RC
formula $\varphi$, we write its translation as a hybrid formula
$(\varphi)'_S$. When translating a formula, $S$ will originally be
empty and it will store pairs of nominals that we will use to simulate
the edges affected by the relation-changing operators we encounter
during the translation.

Intuitively, given that the hybrid operators cannot affect the accessibility relation, we have to 
simulate the updates by recording possible affected edges using nominals and $\down$. 
Notice that as a result, in all the relation-changing logics we will consider, the RC formula 
$\Diamond\psi$ cannot be simply translated into a hybrid formula $\Diamond(\psi)'_S$, even though 
we have $\Diamond$ at our disposition in the hybrid language, because in the source language 
$\Diamond$ is interpreted over the updated accessibility relation.  
Instead, diamond-formulas need to be translated in a way that takes into account the edges
that should be considered deleted, added, or swapped. This is why
the translation of diamond-formulas involve the $\Diamond$ operator
mixed with specific considerations about the set of altered edges $S$.

Consider Sabotage Logic with either the local or global
operator. We use the set $S \subseteq \NOM \times \NOM$ to represent
sabotaged edges, i.e., edges
that have been deleted in a given updated model.

\begin{definition}[Sabotage to Hybrid Logic]\label{transsab}
Let $S \subseteq \NOM \times \NOM$ and $n \in \NOM$.
We define the translation $(~ ~)'_S$ from formulas
of $\bml(\dsab,\dgsab)$ to formulas of $\Hl(\E, \down)$ as:
$$
\begin{array}{rl}
(p)'_S = & p\\
(\neg \varphi)'_S = & \neg (\varphi)'_S\\
(\varphi \wedge \psi)'_S = & (\varphi)'_S \wedge (\psi)'_S\\
(\Diamond \varphi)'_S = & \down n . \lozenge ( \neg {\sf belongs}(n,S) \wedge (\varphi)'_S)\\
(\dsab \varphi)'_S    = & \down n. \lozenge ( \neg {\sf belongs}(n,S) \wedge \down m. (\varphi)'_{S \cup nm})\\
(\dgsab \varphi)'_S     = & \down k. \E \down n. \lozenge ( \neg {\sf belongs}(n,S) \wedge \down m. k{:}(\varphi)'_{S \cup nm})
\end{array}
$$

\noindent where $n$, $m$ and $k$ are nominals that do not appear in $S$, and:
$${\sf belongs}(n,S) = \underset{xy\in S}{\bigvee} (y ~ \wedge ~ n{:}x)$$
\end{definition}

A few comments are in order to understand the translation.
First, given some model $\model=\tup{W,R,V}$ and
some set $S \subseteq \NOM \times \NOM$,
the formula $\down n. \Diamond (\neg {\sf belongs}(n,S))$ is true at
some state $w\in W$ if there exists some state $v$ such that $(w,v) \in R$
and there is no pair of nominals $(x,y) \in S$ such that $(V(x),V(y)) = (w,v)$.
Then, observe that the cases for $\dsab$ and $\dgsab$ modify the set
of deleted pairs in the recursive call to the translation, in both cases by 
adding an edge named $nm$. In the $\dsab$ case, $n$ names the evaluation state
of the formula, while in the $\dgsab$ case, $n$ names some state anywhere in
the model.

Finally, all nominals introduced by the translation are bounded
exactly once. Then we can define the following unequivocal notation:
let $S \subseteq \NOM \times \NOM$,
we define $\bar{S}=\{(\bar{x},\bar{y}) \mid  (x,y) \in S \}$,
where $\bar{n}$ is the state named by the nominal $n \in \NOM$ under
the current valuation of a model.

When considering the translated formula $(\varphi)'_S$ and its truth
in some model $\model=\tup{W,R,V}$, one question that may arise is
what should be the initial valuation of the nominals that appear
in $(\varphi)'_S$. By definition of models for hybrid
logics, nominals must be true at some state. This is not
problematic: in $(\varphi)'_S$, nominals are immediately bounded
by the $\down$ operator, so the truth value of $(\varphi)'_S$ does not depend
on their initial valuation. Hence, we can choose some state $w\in W$
and say that all nominals are bounded to it. This enables us to talk
about equivalence preservation of the translation: the same model $\model$
can be used for $\varphi$ and its translation $(\varphi)'_S$ modulo
the addition of the set of nominals that appear in $(\varphi)'_S$ and their
valuation to some arbitrary state. Then, we can state:

\begin{theorem}
For $\model=\tup{W,R,V}$ a model, $w\in W$, and $\varphi \in \bml(\dsab, \dgsab)$
we have:
$$\model,w\models\varphi ~ ~ \iff ~  ~ \model,w\models(\varphi)'_\emptyset.$$
\end{theorem}

\begin{proof}
We use structural induction on the relation-changing formula, the inductive
hypothesis being:
$$\model^-_{\bar{S}},w\models\varphi ~ ~ \iff ~  ~ \tup{W,R,V'},w\models(\varphi)'_S$$
with $S\subseteq \NOM\times\NOM$, and $V'$ is exactly as $V$
except that for all $(x,y)\in S$, there are $v,u\in W$ such that
$V'(x)=v$ and $V'(y)=u$.
Boolean cases
are straightforward, so we only prove the non-trivial inductive cases. 

\smallskip

\noindent $\varphi=\lozenge\psi$: For the left to right direction, suppose $\model^-_{\bar{S}},w\models\lozenge\psi$.
Then there is some $v\in W$ such that $(w,v)\in R^-_{\bar{S}}$ and $\model^-_{\bar{S}},v\models\psi$.
Because $(w,v)\notin \bar{S}$, then there is no $(x,y)\in S$ such that $(\bar{x},\bar{y})=(w,v)$.
By inductive hypothesis, we have $\model,v\models(\psi)'_S$, and because we can
name $w$ with a fresh nominal $n$, we obtain $\tup{W,R,V^w_n},v\models\neg{\sf belongs}(n,S)\wedge(\psi)'_S$.
Therefore, we have $\model,w\models\down n . \lozenge ( \neg {\sf belongs}(n,S) \wedge (\psi)'_S)$, and
as a consequence we get $\model,w\models(\psi)'_S$.

For the other direction, suppose $\model,w\models(\psi)'_S$, i.e., 
$\model,w\models\down n . \lozenge ( \neg {\sf belongs}(n,S) \wedge (\psi)'_S)$.
Then we have $\tup{W,R,V^w_n},w\models \lozenge ( \neg {\sf belongs}(n,S) \wedge (\psi)'_S)$,
and, by definition, there is some $v\in W$ such that $(w,v)\in R$,
$\tup{W,R,V^w_n},v\models  \neg {\sf belongs}(n,S)$ and $\tup{W,R,V^w_n},v\models (\psi)'_S$.
Because we have $\neg {\sf belongs}(n,S)$, there is no $(x,y)\in S$ such that $(\bar{x},\bar{y})=(w,v)$,
which implies $(w,v)\in R$ if and only if $(w,v)\in R^-_{\bar{S}}$. On the other hand, by inductive hypothesis
we have $\model^-_{\bar{S}},v\models\psi$, then we have $\model^-_{\bar{S}},w\models\lozenge\psi$.

\medskip

\noindent $\varphi=\dsab\psi$: For the left to right direction, suppose $\model^-_{\bar{S}},w\models\dsab\psi$. Then there is some $v\in W$ such that $(w,v)\in R^-_{\bar{S}}$ and $(\model^-_{\bar{S}})^-_{wv},v\models\psi$.
This is equivalent to say $\model^-_{\bar{S}\cup wv},v\models\psi$.
Because $(w,v)\notin \bar{S}$, then there is no $(x,y)\in S$ such that $(\bar{x},\bar{y})=(w,v)$ ~$(\otimes)$.
By inductive hypothesis we have $\tup{W,R,((V')^w_n)^v_m},v\models(\psi)'_{S\cup nm}$, 
where $V'$ is exactly as $V$ but it binds all the nominals which appear in $S$.
By definition, we get $\tup{W,R,(V')^w_n)},v\models\down m. (\psi)'_{S\cup nm}$,
and by $(\otimes)$ we
have $\tup{W,R,(V')^w_n)},v\models \neg{\sf belongs}(n,S) \wedge \down m. (\psi)'_{S\cup nm}$.
Then (by definition) $\tup{W,R,V'},v\models \down n. \lozenge(\neg{\sf belongs}(n,S) \wedge \down m. (\psi)'_{S\cup nm})$,
and, as a consequence, we have $\tup{W,R,V'},v\models (\varphi)'_S$.

For the other direction, suppose $\tup{W,R,V'},w\models(\psi)'_S$, i.e., 
$\tup{W,R,V'},w\models\down n . \lozenge ( \neg {\sf belongs}(n,S) \wedge \down m. (\psi)'_{S\cup nm})$,
where  $V'$ is exactly as $V$ but it binds all the nominals which appear in $S$.
Then, we have $\tup{W,R,(V')^w_n},w\models \lozenge ( \neg {\sf belongs}(n,S) \wedge \down m. (\psi)'_{S\cup nm})$,
and, by definition, there is some $v\in W$ such that $(w,v)\in R$,
$\tup{W,R,V^w_n},v\models  \neg {\sf belongs}(n,S)$ and $\tup{W,R,V^w_n},v\models \down m. (\psi)'_{S\cup nm}$.
Then, $\tup{W,R,((V')^w_n)^v_m},v\models (\psi)'_{S\cup nm}$.
Because we have $\neg {\sf belongs}(n,S)$, there is no $(x,y)\in S$ such that $(\bar{x},\bar{y})=(w,v)$,
which implies $(w,v)\in R$ if and only if $(w,v)\in R^-_{\bar{S}}$.
On the other hand, by inductive hypothesis
we have $\model^-_{\bar{S}\cup wv},v\models\psi$, and thus we have $\model^-_{\bar{S}},w\models\dsab\psi$.

\medskip

\noindent $\varphi=\dgsab\psi$: this case is very similar to the previous one.
\end{proof}

For Bridge Logic, we use the set $B \subseteq \NOM \times \NOM$
to represent the new edges.
New edges present in $B$ mean that the translation of
the modality $\Diamond$ should be able to take them.
This explains why the translation of $\Diamond$ does not
look like a $\Diamond$ with an extra condition, but like an $\E$
with two possibilities: we traverse an edge that is either in the original
model or an edge from the $B$ set.

\begin{definition}[Bridge to Hybrid Logic]\label{transbr}
Let $B \subseteq \NOM \times \NOM$. We define $(~ ~)'_B$ from formulas of \linebreak $\bml(\dbridge,\dgbridge)$
to formulas of $\Hl(\E, \down)$ as:
$$
\begin{array}{rl}
(p)'_B = & p\\
(\neg \varphi)'_B = & \neg (\varphi)'_B\\
(\varphi \wedge \psi)'_B = & (\varphi)'_B \wedge (\psi)'_B\\
(\Diamond \varphi)'_B = &
 \down n . \E \down m ( ~ (n{:}\lozenge m \vee {\sf belongs}(n,B))
 ~ \wedge ~   (\varphi)'_B )\\
(\dbridge \varphi)'_B = & \down n. \E \down m . (
       \neg n{:}\Diamond m \wedge \neg {\sf belongs}(n,B)
               \wedge  (\varphi)'_{B{\cup}nm})\\
(\dgbridge \varphi)'_B = & \down k . \E  \down n. \E \down m . (
      \neg n{:}\Diamond m \wedge \neg {\sf belongs}(n,B)
               \wedge  k{:}(\varphi)'_{B{\cup}nm})\\
\end{array}
$$

\noindent where $n$, $m$ and $k$ are nominals that do not appear in $B$,
and ${\sf belongs}$ is defined as in Definition~\ref{transsab}.
\end{definition}

\begin{theorem}
For $\model=\tup{W,R,V}$ a model, $w\in W$, and $\varphi \in \bml(\dbridge, \dgbridge)$,
we have:
$$\model,w\models\varphi ~ ~ \iff ~  ~ \model,w\models(\varphi)'_\emptyset.$$
\end{theorem}
\begin{proof}
A similar reasoning can be done with the following inductive hypothesis:
$$\model^+_{\bar{B}},w\models\varphi ~ ~ \iff ~  ~ \tup{W,R,V'},w\models(\varphi)'_B$$
with $B\subseteq \NOM\times\NOM$, and $V'$ is exactly as $V$
except that for all $(x,y)\in B$, there are $v,u\in W$ such that
$V'(x)=v$ and $V'(y)=u$.
\end{proof}

We finish with the case of Swap Logic.
We presented a different translation in~\cite{areces14igpl} for the local case only. 
As we did for Sabotage Logic, we use $S \subseteq \NOM \times \NOM$ to
represent the set of deleted edges, i.e., the edges
that should not be possible to traverse in a given updated model.
Indeed, swapping a non-reflexive edge of a model has the effect of
deleting it, along with adding its inverse.
This implies that $S^{-1}$ is a set of edges that we can currently
traverse. All of this requires that $S$ do not contain any reflexive edge,
since a swapped reflexive edge is not deleted. Neither can it contain
a pair of symmetric edges since that would be contradictory.

To ensure this, the translation gets more cautious when handling $\dswap$
and $\dgswap$. When swapping occurs, three possible cases are
taken into account. The first one is when a reflexive edge is swapped.
In that case, the translation continues with the set $S$ left unchanged,
but we require some reflexive
edge to be present, be it at the current state for $\dswap$ with 
$\down n . \Diamond n$, or anywhere in the model for $\dgswap$
with $\E\down n . \Diamond n$.

The second case is when we swap an irreflexive edge
that has never been swapped before.
Hence we ensure that this edge is present in the model, that it is
irreflexive, and that neither this edge nor its
inverse is in $S$. We then add the nominals that name it to $S$
before moving on with the translation.

The last case is when we traverse an already swapped edge.
That is, for some $xy \in S$, we traverse the edge referred to
by the nominals $yx$. In this case, we do not need to require the presence
of any new edge in the model. We assume to be standing at the state named
by $y$ and that the rest of the formula is satisfied at $x$, with the
modification that we remove $xy$ from $S$ and add $yx$ to it.

An attentive reader would object: why not just remove $xy$ from the set
$S$ since swapping some edge twice just makes it return to its configuration
in the original model? The answer is that there is a corner case when some edge
\emph{and} its symmetric are both present in the initial model. Then, the action
of swapping it twice is not supposed to restore its symmetric. This is what we do
by adding $yx$ to the set $S$: we ensure the former symmetric edge is no
longer present.

\begin{definition}[Swap to Hybrid Logic]\label{transswap}
Let $S \subseteq \NOM \times \NOM$. 
We define $(~ ~)'_S$ from formulas of \linebreak$\bml(\dswap,\dgswap)$ to formulas of
$\Hl(\E, \down)$ as:
$$
\begin{array}{rl}
(p)'_S = & p\\
(\neg \varphi)'_S = & \neg (\varphi)'_S\\
(\varphi \wedge \psi)'_S = & (\varphi)'_S \wedge (\psi)'_S\\
(\Diamond \varphi)'_S = &
 (\down n . \lozenge (\neg{\sf belongs}(n,S) \wedge (\varphi)'_S ))
 ~ ~\vee ~ ~ {\sf isSat}(S^{-1},(\varphi)'_S) \\
(\dswap \varphi)'_S = &
  ~ ~ ~ ( ~ \down n. \Diamond n ~ ~ \wedge ~ ~ (\varphi)'_S ~)\\
 & \vee ~ \down n . \lozenge ({\neg}n
               \wedge \neg{\sf belongs}(n,S)
               \wedge \neg{\sf belongs}(n,S^{-1})
               \wedge \down m . (\varphi)'_{S\cup nm})\\
 & \vee ~  \underset{xy \in S}{\bigvee} (y \wedge x{:}(\varphi)'_{(S\setminus xy) \cup yx}) \\
(\dgswap \varphi)'_S = &
   ~ ~ ~ (~ \E \down n. \Diamond n ~ ~ \wedge ~ ~ (\varphi)'_S ~ )\\
 & \vee ~ \down k. \E \down n . \lozenge ({\neg}n
               \wedge \neg{\sf belongs}(n,S)
               \wedge \neg{\sf belongs}(n,S^{-1})
               \wedge \down m. k{:}(\varphi)'_{S\cup nm}) \\
 & \vee ~  \underset{xy \in S}{\bigvee} (\varphi)'_{(S\setminus xy) \cup yx} \\
\end{array}
$$
\noindent where $n$, $m$ and $k$ are nominals that do not appear in $S$,
${\sf belongs}$ is defined as in Definition~\ref{transsab}, and
$$
{\sf isSat}(S,\varphi) =  \underset{xy\in S}{\bigvee} (x \wedge y{:}\varphi).
$$
\end{definition}

The formula ${\sf isSat}(S,(\varphi)'_S)$ says that the translation of $\varphi$ is
satisfiable at the end of some of the edges belonging to $S$. Note that
the translation maps formulas of $\bml(\dswap)$ to the less expressive
$\Hl(:, \down)$, i.e., the $\E$ operator is not required.

\begin{theorem}
For $\model=\tup{W,R,V}$ a model, $w\in W$ and $\varphi \in \bml(\dswap, \dgswap)$
we have:
$$\model,w\models\varphi ~ ~ \iff ~  ~ \model,w\models(\varphi)'_\emptyset.$$
\end{theorem}
\begin{proof}
Again, a similar reasoning can be done with the following inductive hypothesis:
$$\model^*_{\bar{S^{-1}~}},w\models\varphi ~ ~ \iff ~  ~ \tup{W,R,V'},w\models(\varphi)'_S$$
with $S\subseteq \NOM\times\NOM$, and $V'$ is exactly as $V$
except that for all $(x,y)\in S$, there are $v,u\in W$ such that
$V'(x)=v$ and $V'(y)=u$.
\end{proof}


%% file: implementation.tex
\section{Implementation and Examples}
\label{sec:imple}

We have implemented these translations as a new feature of the
tableaux-based theorem prover HTab~\cite{Hoffmann2007}. Its version
1.7.1 can be downloaded from \url{http://hub.darcs.net/gh/htab} along with
example formulas.
Instructions are provided, explaining how
to use HTab to check satisfiability of some relation-changing formula
and to generate a model from an open tableau.

HTab originally handles the hybrid logic $\Hl(:, \E, \down)$,
and guarantees termination
of any fragment that lacks the $\down$ binder.
As part of the work presented in this article, we added the following
feature: when passed the \verb_--translate_ flag, HTab interprets the
input formula as a relation-changing one. It first translates it to
the corresponding $\Hl(\E, \down)$ formula (or more precisely,
$\Hl(:,\down)$ formula in the case of local sabotage and local swap),
and then runs its internal hybrid tableaux calculus on the translation.

Since the translation uses the $\down$ binder, HTab may never terminate
on some specific relation-changing formulas. Even in the terminating cases,
the size of the translated formula (in particular for
swap logic) can make HTab run for a very long time before giving an answer.
However, there are several possible workarounds.
First, a time limit in seconds can be given with
the flag \verb_-t_. Also, it may be useful to disable the semantic branching
optimization by passing \verb_--sembranch=no_. Tweaking unit propagation
can also be useful for certain formulas, in some cases by disabling it
(with \verb_--no-unit-prop_), and in others, by making it more aggressive (with \verb_--eager_).
The flag \verb_--minimal_ makes HTab try reusing existing states in the
model instead of systematically generating new ones. This introduces
many more branches in the tableau, making it slower, but it can be crucial
to make some cases terminate. Finally, the flag \verb_--random_
uses pseudorandomness to select the next disjunct formula
to expand in the tableau, and also shuffles the order in which
the branches are explored, including those introduced
by \verb_--minimal_. The advantage is that some pseudorandom run of HTab
on a given input formula could terminate quickly by doing the right choices.
Then, this run can be reproduced by setting the pseudorandom seed of HTab
with the \verb_--seed_ flag.
More information is available by running \verb_htab --help_.

For all three translations, the implementation has the following particular case:
$$(\Diamond \varphi)'_\emptyset = \Diamond (\varphi)'_\emptyset$$

This avoids introducing unnecessary nominals and makes the translated formula
more readable. The generated hybrid formula can be seen by using the \verb_--showformula_ flag.

Since the translations are equivalence-preserving, the models built by HTab satisfy the
input RC formula. Let us see a few examples.
We present the formulas in multi-line way, each line is a conjunct of the whole
formula. Also, some new logical connectors are used in the expected way.

\bigskip

\begin{minipage}[t]{0.3\textwidth}
Input $\bml{\dsab}$ formula:

$$
\begin{array}{l}
\Diamond(A \wedge \neg B  \wedge \Diamond\Diamond A)\\
\Diamond(B \wedge \neg A  \wedge \Diamond\Diamond B)\\
\bsab(A \rightarrow \Box\Box\neg A)\\
\bsab(B \rightarrow \Box\Box\neg B)\\
\end{array}
$$
\end{minipage}
~
\begin{minipage}[t]{0.3\textwidth}
Translated hybrid formula:

\vspace{-0.4cm}
$$
\begin{array}{l}
\Diamond(A \wedge \neg B  \wedge \Diamond\Diamond A)\\
\Diamond(B \wedge \neg A  \wedge \Diamond\Diamond B)\\
\down n_0 . \Box ( \down n_1  . (\neg A \\
~ ~ ~  \vee \down n_2 . \Box ( ( n_1 \wedge n_2: n_0)\\
~ ~ ~ ~ ~  \vee \down n_3 . \Box ( (n_1 \wedge n_3:n_0 ) \\
~ ~ ~ ~ ~ ~ ~  \vee \neg A ) ))  )\\
\down n_0 . \Box ( \down n_1  . (\neg B \\
~ ~ ~  \vee \down n_2 . \Box ( ( n_1 \wedge n_2: n_0)\\
~ ~ ~ ~ ~  \vee \down n_3 . \Box ( (n_1 \wedge n_3:n_0 ) \\
~ ~ ~ ~ ~ ~ ~  \vee \neg B ) ))  )\\
\end{array}
$$
\end{minipage}
~
\begin{minipage}[t]{0.3\textwidth}
Model found by HTab:

\vspace{-0.4cm}
\begin{center}
\includegraphics[width=3cm]{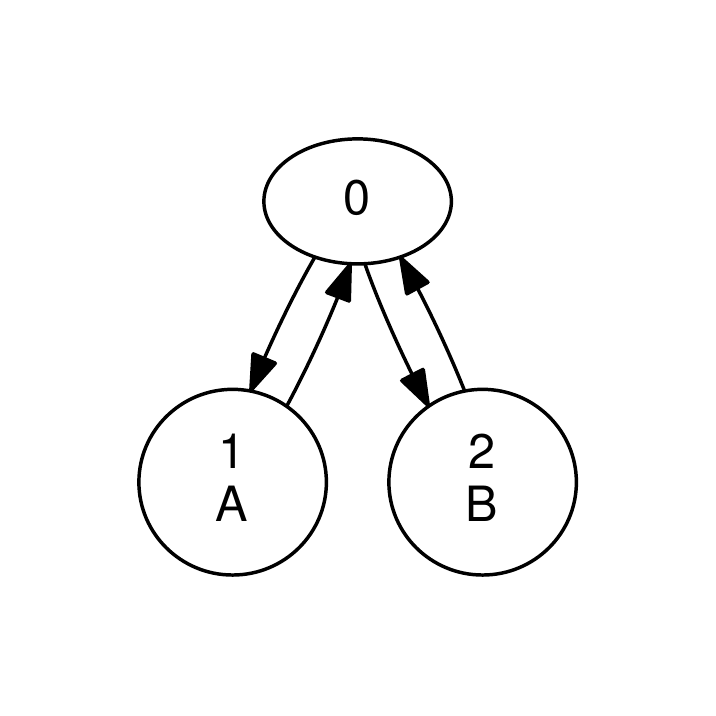}
\end{center}
\end{minipage}

\bigskip

\begin{minipage}[t]{0.3\textwidth}
Input $\bml{\dgsab}$ formula:

$$
\begin{array}{l}
\Diamond(A \wedge \neg B  \wedge \Diamond\Diamond A)\\
\Diamond(B \wedge \neg A  \wedge \Diamond\Diamond B)\\
\Box\Box(C \wedge \Box\neg C)\\
\dgsab\Box\Box\Box \bot \\
\end{array}
$$
\end{minipage}
~
\begin{minipage}[t]{0.3\textwidth}
Translated hybrid formula:

\vspace{-0.4cm}
$$\begin{array}{l}
\Diamond(A \wedge \neg B  \wedge \Diamond\Diamond A)\\
\Diamond(B \wedge \neg A  \wedge \Diamond\Diamond B)\\
\Box\Box(C \wedge \Box\neg C)\\
\down n_0 . \E \down n_1
  \Diamond( \down n_2 .\\
~ ~ ~  n_0:(\down n_3 . \Box (
   ( n_2 \wedge n_3:n_1)  \\
~ ~ ~ ~ ~ \vee ~ \down n_4 . \Box (
     (n_2 \wedge n_4:n_1) \\
~ ~ ~ ~ ~ ~ ~ \vee ~ \down n_5 . \Box (
       (n_2 \wedge n_5:n_1) ) ) ) ) ) \\
\end{array}
$$
\end{minipage}
~
\begin{minipage}[t]{0.3\textwidth}
Model found by HTab:

\vspace{-0.6cm}
\begin{center}
\includegraphics[width=3cm,height=4.6cm,keepaspectratio=false]{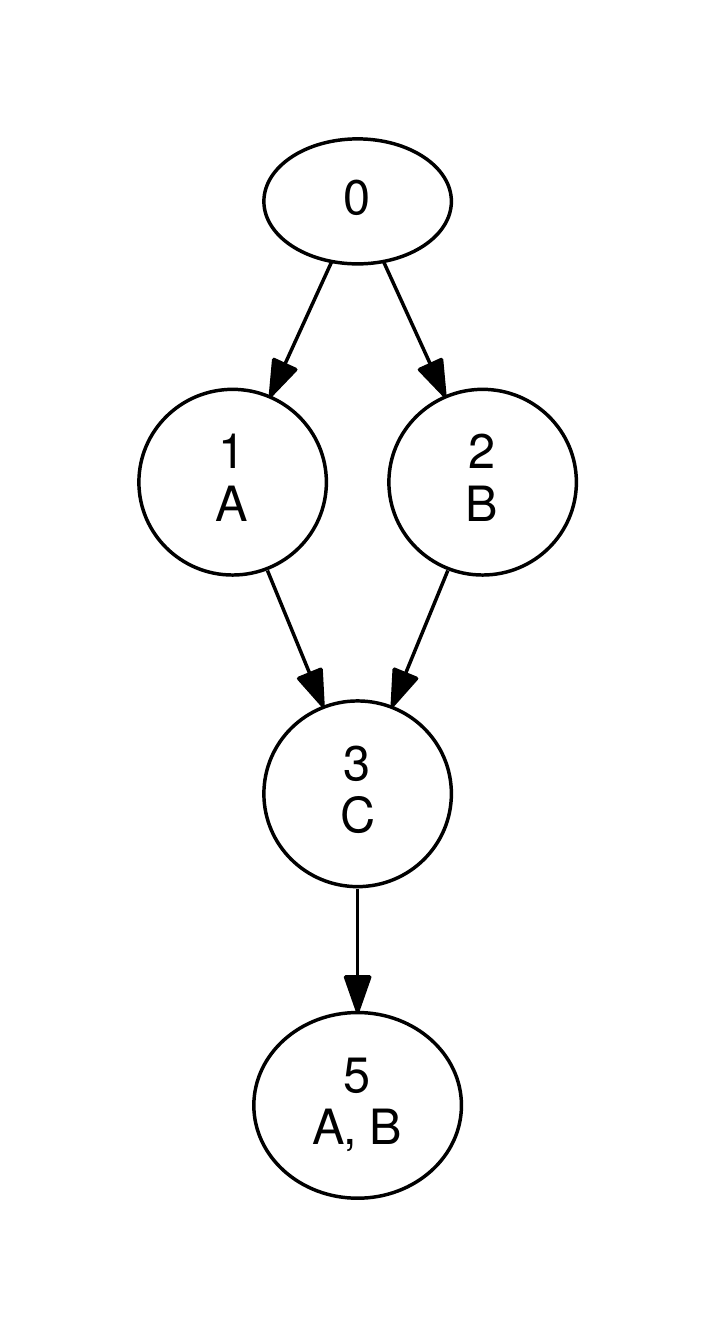}
\end{center}
\end{minipage}

\bigskip

\begin{minipage}[t]{0.3\textwidth}
Input $\bml{\dswap}$ formula:

\vspace{-0.5cm}
$$
\begin{array}{l}
\Diamond(A \wedge \neg B)\\
\Diamond(B \wedge \neg A)\\
\Box\Diamond \top\\
\Box\Box\Box \bot\\
\bswap\Box\bswap\Box\Box \bot\\
\Box\bswap\Box\Box \bot\\
\dswap\dswap\Diamond\Diamond\Diamond\Diamond\Diamond \top\\
\end{array}
$$
\end{minipage}
~
\begin{minipage}[t]{0.3\textwidth}
Translated hybrid formula:

\bigskip

\bigskip

\bigskip
\emph{\ldots too big to be displayed \ldots}
\end{minipage}
~
\begin{minipage}[t]{0.3\textwidth}
Model found by HTab:

\vspace{-0.4cm}
\begin{center}
\includegraphics[width=3cm]{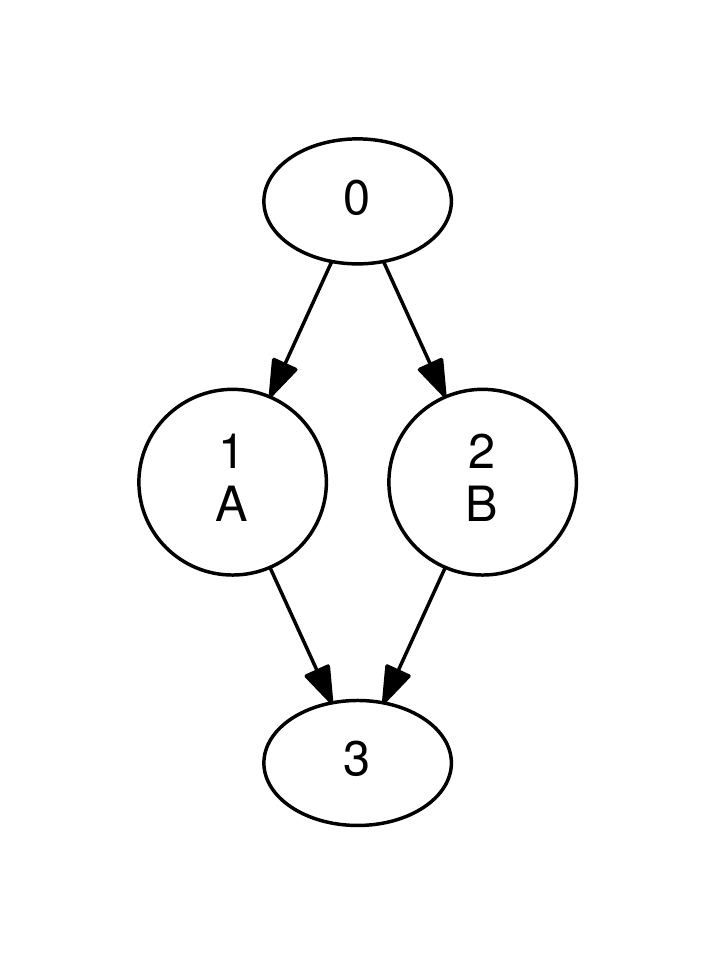}
\end{center}
\end{minipage}

More formulas from the six RC logics are available in the \verb_./rc/_
directory of the HTab source code.
This implementation is useful to check the correctness of the translations,
just by checking the satisfiable/unsatisfiable output of the prover for known
RC formulas. It is
also useful for checking that RC formulas build models in the expected way,
such as non-tree or diamond-shaped models.


%% file: decidable.tex
\section{Decidable Fragments}\label{sec:decid}

Interesting decidable fragments of hybrid logics with binders have been
found over time.
Such decidable fragments are convenient for our relation-changing logics
in the light of the (computable) translations presented in Section~\ref{sec:trans}.
First, let us consider restricting the satisfiability problem over certain classes
of models. The following logics are known to be decidable over the indicated classes:

\begin{itemize}
\item[-] $\Hl(\E,\down)$ over linear frames (i.e., irreflexive, transitive, and trichotomous frames) \cite{FranceschetRS03,schneider07phd} (this includes $(\mathbb{N},<)$),
\item[-] $\Hl(\E,\down)$ over models with a single, transitive tree relation \cite{schneider07phd},
\item[-] $\Hl(\E,\down)$ over models with a single, $S5$, or complete relation \cite{schneider07phd},
\item[-] $\Hl(:,\down)$ over models with a single relation of bounded finite width~\cite{tencate2005complexity}; 
   as a corollary, also over finite models.
\end{itemize}

Since the translations preserve equivalence, we get:

\begin{corollary} \label{coro:semsat}
The satisfiability problem for all relation-changing modal logics over
linear, transitive trees, $S5$, and complete frames is decidable.
\end{corollary}

\begin{corollary}
The satisfiability problem for local sabotage and local swap logics over
models of bounded width is decidable.
\end{corollary}

Curiously, these results mean that relation-changing modal logics are
decidable over certain classes of models, even if the modifications
implied by evaluating RC formulas yield models that \emph{do not}
belong to such class. For instance, these two facts are simultaneously true:
sabotage logic is decidable on the class of $S5$ models, and deleting edges in
an $S5$ model can yield a non-$S5$ model.

Now, let us turn to syntactical definitions of decidable fragments.
We recall that local sabotage and local swap can be translated to $\Hl(:,\down)$.
Consider formulas of $\Hl(:,\down)$ in negation normal form.
$\Hl(:,\down) \setminus  \Box\down\Box$ is the fragment obtained by removing
formulas that contain a nesting of $\Box$, $\down$ and again $\Box$.
This fragment is decidable \cite{tencate2005complexity}.

Our translations do use the $\down$ binder in many places, but we can make
them a little more economical in that sense, at the expense of losing succinctness.

Take the following case for $\bml(\dswap)$:
$$
\begin{array}{rl}
(\Diamond \varphi)'_S = & \down n . \lozenge ( \neg {\sf belongs}(n,S) \wedge (\varphi)'_S).\\
\end{array}
$$
Instead of using the down-arrow binder and later ensuring that we
did not take a deleted edge by using $\neg {\sf belongs}(n , S)$, we can
do the following. For all pairs of nominals $(x,y) \in S$, the current state
$w$ satisfies one combination of the truth values of the nominals $x$. Let $X$ be
the set of true nominals $x$ at $w$.
Then, $(\varphi)'_S$ should be true at some accessible state $v$ that should not
satisfy any of the corresponding $y$ nominals for all $x \in X$.

Then, the translation becomes:
$$
\begin{array}{rl}
(\Diamond \varphi)'_S =
 & \underset{X \subseteq {\sf fst(S)}}{\bigvee}(
    \underset{x \in X}{\bigwedge} x \wedge \underset{x \notin X}\bigwedge \neg x
       \wedge \Diamond
         (\underset{y \in {\sf snd}(S,X)}{\bigwedge} \neg y \wedge (\varphi)'_S))
\end{array}
$$
where ${\sf fst}(S) = \{ x \mid (x,y) \in S \}$
and ${\sf snd}(S,X) = \{ y \mid (x,y) \in S , x \in X \}$.

In the case of $\bml(\dswap)$ we can do the same. We recall that the case
introduced in Section~\ref{sec:trans} was:
$$
\begin{array}{rl}
(\Diamond \varphi)'_S = &
 (\down n . \lozenge (\neg{\sf belongs}(n,S) \wedge (\varphi)'_S ))
 ~ ~\vee ~ ~ {\sf isSat}(S^{-1},(\varphi)'_S). \\
\end{array}
$$
Here the ${\sf isSat}(S^{-1},(\varphi)'_S)$ disjunct does not use
the $\down$ binder, while the first disjunct is similar to the case
of local sabotage, and can be replaced accordingly:
$$
\begin{array}{rl}
(\Diamond \varphi)'_S =
 & ~ ~ ~ ~ ~ \underset{X \subseteq {\sf fst(S)}}{\bigvee}(
    \underset{x \in X}{\bigwedge} x \wedge \underset{x \notin X}\bigwedge \neg x
       \wedge \Diamond
         (\underset{y \in {\sf snd}(S,X)}{\bigwedge} \neg y \wedge (\varphi)'_S))\\
 & \vee  ~ ~ {\sf isSat}(S^{-1},(\varphi)'_S).\\
\end{array}
$$
Let $\blacklozenge$ be either $\dsab$ or $\dswap$
and $\blacksquare$ be either $\bsab$ or $\bswap$. The following patterns in RC formulas
provoke the following patterns in the hybrid formula produced by the
translations:
\begin{center}
\begin{tabular}{|c |c|}
\hline
RC pattern & Produced pattern \\
\hline
$\Box$ & $\Box$ \\
\hline 
$\blacklozenge$ & $\down$ \\
\hline 
$\blacksquare$ &  $\down \Box \down$\\
\hline 
\end{tabular}
\end{center}

By considering these new versions of the translations, and by taking into account the syntactic decidable fragment of $\Hl(:,\down)$ mentioned above, we can establish the following result:

\begin{corollary} \label{th:syxsat}
The following fragments are decidable on the class of all relational models:
\begin{itemize}
\item[-] $\bml(\dsab) \setminus \{ \blacksquare\blacksquare,
                                \blacksquare\Box,
                                \Box\blacksquare,
                                \greysquare\blacklozenge\greysquare \}$
\item[-] $\bml(\dswap)\setminus \{ \blacksquare\blacksquare,
                                \blacksquare\Box,
                                \Box\blacksquare,
                                \greysquare\blacklozenge\greysquare \}$
\end{itemize}

\noindent where $\greysquare$ is either $\Box$ or $\blacksquare$.
\end{corollary}

%% file: exppow.tex
\section{Comparing Expressive Power} 
\label{sec:exppow}

We have introduced translations for the six relation-changing modal logics
from Section~\ref{sec:basic} into hybrid logic. In some cases (for the local version
of swap and sabotage), the obtained formulas fall into the fragment $\Hl(:,\down)$.
On the other hand, for encoding the rest of the logics we need also to use the
universal modality $\E$. An interesting question is whether we can obtain translations
from hybrid to relation-changing logics, i.e., if some of the relation-changing
logics considered in this article are as expressive as some hybrid logic. Let us define
first, the expressive power comparisons we will use.

\begin{definition}[${\cal L} \leq {\cal L}'$] 
\label{def:expcomp}
We say that ${\cal L}'$ is \emph{at least as expressive as} ${\cal L}$
(notation ${\cal L}\leq{\cal L}'$) if there is a function $\Tr$ between
formulas of ${\cal L}$ and ${\cal L}'$ such that for
every model $\model$ and every formula $\varphi$ of ${\cal L}$ we have that
$$\model \models_{\cal L} \varphi \mbox{ iff } \model \models_{\cal L'} \Tr(\varphi).$$
$\model$ is seen as a model of
${\cal L}$ on the left and as a model of ${\cal L'}$ on the right, and we
use in each case the appropriate semantic relation $\models_{\cal L}$ or
$\models_{\cal L'}$ as required.

$\L'$ is strictly more expressive than $\L$ ($\L<\L'$) if $\L\leq\L'$ but
not $\L'\leq\L$. Finally, we say that ${\cal L}$ and ${\cal L}'$ are
\emph{incomparable} if ${\cal L} \nleq {\cal L}'$ and  ${\cal L}' \nleq {\cal L}$.
\end{definition}



In~\cite{areces12,areces14igpl,fervari14phd,AFH15} we discussed the expressive power
of relation-changing modal logics by introducing their corresponding notions
of bisimulations and using them to compare the logics among each other. We concluded that
they are all {\em incomparable in expressive power.}\footnote{Except for the 
local and global swap operators, which is still open in one direction.} As a consequence, we conclude
that it is not possible that two of them capture the same fragment of hybrid logic.
In fact, we will prove that all the relation-changing logics considered here
are strictly less expressive than the corresponding hybrid logic in which they are translated.

\begin{theorem}
\label{th:rcstricthl}
Let $\blacklozenge_1\in\{\dsab,\dswap\}$, we have $\bml(\blacklozenge_1)<\Hl(:,\down)$.
For $\blacklozenge_2\in\{\dgsab,\dgswap,\dbridge,\dgbridge\}$, we have 
$\bml(\blacklozenge_2)<\Hl(\E,\down)$.
\end{theorem}

\begin{proof}
For any of the logics mentioned above, we have translations into the 
corresponding hybrid logic. Now we need to prove that these translations
do not cover their entire target language (modulo equivalence). In order to do that, 
we provide bisimilar models for relation-changing modal logics which can
be distinguished by some hybrid formula. In Figure~\ref{fig:models}, we show
two pairs of models already introduced in~\cite{AFH15} that cover all possibilities of bisimilarity.

\begin{figure}[!h]
\begin{center}
\begin{tabular}{|c |c|l|}
\hline
$\model,w$ & $\model',w'$ & Bisimilar for\\
\hline
\hline 
\begin{minipage}{.25\linewidth}
\begin{tikzpicture}[>=latex]
\node (b) at (3,0)   [shape=circle,draw,fill, inner sep=1pt, label=below:$ w$] {} ;
\path [ar,->] (b) edge [loop] (b);
\end{tikzpicture}
\end{minipage}
&
\begin{minipage}{.25\linewidth}
\begin{tikzpicture}[>=latex]
\node (a0) at (7,0)   [shape=circle,draw,fill, inner sep=1pt,label=below:$ w'$] {} ;
\node (a1) at (9,0)  [shape=circle,draw,fill, inner sep=1pt] {} ;
\path [ar,->] (a0) edge [loop] (a0);
\path [ar,->] (a1) edge [loop] (a1);
\draw [ar,->] (a0) edge [bend left] (a1);
\draw [ar,->] (a1) edge [bend left] (a0);
\end{tikzpicture}
\end{minipage}
& \begin{minipage}{2cm}
\smallskip
  $\bml(\dswap)$
  $\bml(\dbridge)$
  $\bml(\dgswap)$
  $\bml(\dgbridge)$
  \end{minipage}
\\
\hline
\begin{minipage}{.25\linewidth}
\begin{tikzpicture}[>=latex]
\node (b0) at (2,0)   [shape=circle,draw,fill, inner sep=1pt, label=above:$ w$] {} ;
\node (b1) at (4,0)   [shape=circle,draw,fill, inner sep=1pt, label=above:$ $] {} ;
\node (b2) at (3,-1)   [shape=circle,draw,fill, inner sep=1pt, label=below:$ $] {} ;
\draw [ar,->] (b0) edge [left] (b1);
\draw [ar,->] (b1) edge [left] (b2);
\draw [ar,->] (b0) edge [left] (b2);
\end{tikzpicture}
\end{minipage}
&
\begin{minipage}{.25\linewidth}
\begin{tikzpicture}[>=latex]
\node (a0) at (6,0)   [shape=circle,draw,fill, inner sep=1pt, label=above:$ w'$] {} ;
\node (a1) at (8,0)  [shape=circle,draw,fill, inner sep=1pt,  label=above:$ $] {} ;
\node (a2) at (8,-1)  [shape=circle,draw,fill, inner sep=1pt,  label=below:$ $] {} ;
\node (a3) at (6,-1)  [shape=circle,draw,fill, inner sep=1pt,  label=below:$ $] {} ;
\draw [ar,->] (a0) edge [left] (a1);
\draw [ar,->] (a1) edge [left] (a2);
\draw [ar,->] (a0) edge [left] (a3);
\end{tikzpicture}
\end{minipage}
& \begin{minipage}{2cm}
  $\bml(\dsab)$\\
  $\bml(\dgsab)$
  \end{minipage}
  \\
\hline
\end{tabular}
\end{center}
\caption{Bisimilar models}\label{fig:models}
\end{figure}

The two models in the first row can be distinguished by
the formula $\down n. \Box n$, which establishes that the only successor
of the evaluation point is itself. This formula is true at $\model,w$
and false at $\model',w'$. Models in the second row can be distinguished
by the formula $\down n. \lozenge \down m. n{:}\lozenge\lozenge m$,
which says that from the evaluation point it is possible to arrive to the same
state in one or two steps. This is true at
$\model,w$ but false at $\model',w'$.
\end{proof}

Notice that both hybrid formulas we introduced above belong
to the fragment $\Hl(:,\down)$, i.e., it was not necessary to use the $\E$ operator.
This means that even though we use $\E$ in some of the translations (and we strongly believe
that it is essential for some encodings) there are fragments
of $\Hl(:,\down)$ that cannot be captured by relation-changing modal logics.

%% file: final.tex
\section{Final Remarks}
\label{sec:final}

In this article, we introduced equivalence-preserving translations from six logics we named
{\em relation-changing} to a very expressive hybrid logic.
We considered three kinds of modifications:
deleting, adding, and swapping edges, that can be performed both globally (anywhere in the model) 
and locally (modifying adjacent edges from the evaluation point). 
On the other hand, hybrid logic has operators to rename states
in a model with some particular atomic symbols named nominals. We use the down-arrow operator
$\down$ to name pairs of states that represent modified states. In this way, we keep
track of the evolution of a model.

It is known that the hybrid logic $\Hl(\E,\down)$ has the same expressive power as $\fol$,
and we introduced standard translations from relation-changing logics to $\fol$ in~\cite{AFH15}.
However, by giving explicit translations to hybrid logic we can benefit from its decidable
fragments to find decidable fragments of relation-changing
modal logics. Also, these translations are useful to analyze expressive power. We showed that
the six logics we considered are strictly less expressive than $\Hl(\E,\down)$. In fact,
despite we used the modality $\E$ in some translations, all relation-changing logics we
considered here cannot capture the full fragment $\Hl(:,\down)$ (which is less expressive than
$\Hl(\E,\down)$). In summary, we learned that relation-changing modal logics are languages
that enable to talk directly and succinctly about distinct kinds of model modifications,
but with a little effort they can be simulated by hybrid logics.

Translating to $\Hl(\E, \down)$ also enabled us to easily obtain an implementation
for relation-changing modal logics, by extending the hybrid logic theorem prover HTab.
Satisfiability checking and model building can thus be automated and were useful to
empirically verify our translations on concrete cases.
However, we did not implement the changes of Section~\ref{sec:decid}.
Indeed, although in some cases they provide a way to avoid using the down-arrow binder in
the translations (which is a source of undecidability for hybrid logic),
they also make the generated formulas impractically bigger.

We studied six relation-changing modal logics with the goal of covering a sufficiently varied sample 
of alternatives. Clearly, other operators could have been included in this exploration, and actually 
some alternative choices have been already investigated in the literature, e.g., the adjacent sabotage 
operator discussed in~\cite{rohde06phd}, or the more generic approach investigated in~\cite{AFH15}. 

Even though relation-changing modal logics have been extensively investigated~\cite{fervari14phd,AFH15}, 
there are still many interesting questions to be answered. The {\em hybrid perspective} we introduced in this article,
as well as in~\cite{areces13}, gives us a new way to think of the relation-changing framework.
As an example, we can use {\em hybridization techniques} (a very standard technique in modal logic~\cite{blackburn01})
to find complete axiomatizations or compute interpolants for relation-changing modal logics.

%% file: gandalf16.bbl
\begin{thebibliography}{10}
\providecommand{\bibitemdeclare}[2]{}
\providecommand{\surnamestart}{}
\providecommand{\surnameend}{}
\providecommand{\urlprefix}{Available at }
\providecommand{\url}[1]{\texttt{#1}}
\providecommand{\href}[2]{\texttt{#2}}
\providecommand{\urlalt}[2]{\href{#1}{#2}}
\providecommand{\doi}[1]{doi:\urlalt{http://dx.doi.org/#1}{#1}}
\providecommand{\bibinfo}[2]{#2}

\bibitemdeclare{inproceedings}{arecroad99}
\bibitem{arecroad99}
\bibinfo{author}{C.~\surnamestart Areces\surnameend},
  \bibinfo{author}{P.~\surnamestart Blackburn\surnameend} \&
  \bibinfo{author}{M.~\surnamestart Marx\surnameend} (\bibinfo{year}{1999}):
  \emph{\bibinfo{title}{A Road-map on Complexity for Hybrid Logics}}.
\newblock In \bibinfo{editor}{J.~\surnamestart Flum\surnameend} \&
  \bibinfo{editor}{M.~\surnamestart Rodr{\'i}guez-Artalejo\surnameend},
  editors: {\sl \bibinfo{booktitle}{Computer Science Logic}}, {\sl
  \bibinfo{series}{Lecture Notes in Computer Science}} \bibinfo{volume}{1683},
  \bibinfo{publisher}{Springer}, \bibinfo{address}{Madrid, Spain}, pp.
  \bibinfo{pages}{307--321}, \doi{10.1007/3-540-48168-0\_22}.

\bibitemdeclare{incollection}{arec:hybr05b}
\bibitem{arec:hybr05b}
\bibinfo{author}{C.~\surnamestart Areces\surnameend} \&
  \bibinfo{author}{B.~\surnamestart ten Cate\surnameend}
  (\bibinfo{year}{2007}): \emph{\bibinfo{title}{Hybrid {L}ogics}}.
\newblock In \bibinfo{editor}{P.~\surnamestart Blackburn\surnameend},
  \bibinfo{editor}{F.~\surnamestart Wolter\surnameend} \&
  \bibinfo{editor}{J.~\surnamestart van Benthem\surnameend}, editors: {\sl
  \bibinfo{booktitle}{Handbook of Modal Logic}}, \bibinfo{publisher}{Elsevier},
  pp. \bibinfo{pages}{821--868}, \doi{10.1016/s1570-2464(07)80017-6}.

\bibitemdeclare{inproceedings}{areces14wollic}
\bibitem{areces14wollic}
\bibinfo{author}{C.~\surnamestart Areces\surnameend},
  \bibinfo{author}{H.~\surnamestart van Ditmarsch\surnameend},
  \bibinfo{author}{R.~\surnamestart Fervari\surnameend} \&
  \bibinfo{author}{F.~\surnamestart Schwarzentruber\surnameend}
  (\bibinfo{year}{2014}): \emph{\bibinfo{title}{Logics with {C}opy and
  {R}emove}}.
\newblock In: {\sl \bibinfo{booktitle}{Logic, Language, Information, and
  Computation}}, {\sl \bibinfo{series}{Lecture Notes in Computer Science}}
  \bibinfo{volume}{8652}, \bibinfo{publisher}{Springer}, pp.
  \bibinfo{pages}{51--65}, \doi{10.1007/978-3-662-44145-9\_4}.

\bibitemdeclare{article}{ADFS15}
\bibitem{ADFS15}
\bibinfo{author}{C.~\surnamestart Areces\surnameend},
  \bibinfo{author}{H.~\surnamestart van Ditmarsch\surnameend},
  \bibinfo{author}{R.~\surnamestart Fervari\surnameend} \&
  \bibinfo{author}{F.~\surnamestart Schwarzentruber\surnameend}
  (\bibinfo{year}{2015}): \emph{\bibinfo{title}{The {M}odal {L}ogic of {C}opy
  and {R}emove}}.
\newblock {\sl \bibinfo{journal}{To Appear in Information and Computation,
  special issue of WoLLIC 2014}}.

\bibitemdeclare{incollection}{areces12}
\bibitem{areces12}
\bibinfo{author}{C.~\surnamestart Areces\surnameend},
  \bibinfo{author}{R.~\surnamestart Fervari\surnameend} \&
  \bibinfo{author}{G.~\surnamestart Hoffmann\surnameend}
  (\bibinfo{year}{2012}): \emph{\bibinfo{title}{Moving {A}rrows and {F}our
  {M}odel {C}hecking {R}esults}}.
\newblock In: {\sl \bibinfo{booktitle}{Logic, Language, Information and
  Computation}}, {\sl \bibinfo{series}{Lecture Notes in Computer Science}}
  \bibinfo{volume}{7456}, \bibinfo{publisher}{Springer}, pp.
  \bibinfo{pages}{142--153}, \doi{10.1007/978-3-642-32621-9\_11}.

\bibitemdeclare{inproceedings}{areces13}
\bibitem{areces13}
\bibinfo{author}{C.~\surnamestart Areces\surnameend},
  \bibinfo{author}{R.~\surnamestart Fervari\surnameend} \&
  \bibinfo{author}{G.~\surnamestart Hoffmann\surnameend}
  (\bibinfo{year}{2013}): \emph{\bibinfo{title}{Tableaux for
  {R}elation-{C}hanging {M}odal {L}ogics}}.
\newblock In: {\sl \bibinfo{booktitle}{Frontiers of Combining Systems}}, {\sl
  \bibinfo{series}{Lecture Notes in Computer Science}} \bibinfo{volume}{8152},
  pp. \bibinfo{pages}{263--278}, \doi{10.1007/978-3-642-40885-4\_19}.

\bibitemdeclare{article}{areces14igpl}
\bibitem{areces14igpl}
\bibinfo{author}{C.~\surnamestart Areces\surnameend},
  \bibinfo{author}{R.~\surnamestart Fervari\surnameend} \&
  \bibinfo{author}{G.~\surnamestart Hoffmann\surnameend}
  (\bibinfo{year}{2014}): \emph{\bibinfo{title}{Swap {L}ogic}}.
\newblock {\sl \bibinfo{journal}{Logic Journal of the IGPL}}
  \bibinfo{volume}{22}(\bibinfo{number}{2}), pp. \bibinfo{pages}{309--332},
  \doi{10.1093/jigpal/jzt030}.

\bibitemdeclare{article}{AFH15}
\bibitem{AFH15}
\bibinfo{author}{C.~\surnamestart Areces\surnameend},
  \bibinfo{author}{R.~\surnamestart Fervari\surnameend} \&
  \bibinfo{author}{G.~\surnamestart Hoffmann\surnameend}
  (\bibinfo{year}{2015}): \emph{\bibinfo{title}{Relation-{C}hanging {M}odal
  {O}perators}}.
\newblock {\sl \bibinfo{journal}{Logic Journal of the {IGPL}}}
  \bibinfo{volume}{23}(\bibinfo{number}{4}), pp. \bibinfo{pages}{601--627},
  \doi{10.1093/jigpal/jzv020}.

\bibitemdeclare{inproceedings}{vanbenthem05}
\bibitem{vanbenthem05}
\bibinfo{author}{J.~\surnamestart van Benthem\surnameend}
  (\bibinfo{year}{2005}): \emph{\bibinfo{title}{An {E}ssay on {S}abotage and
  {O}bstruction}}.
\newblock In: {\sl \bibinfo{booktitle}{Mechanizing Mathematical Reasoning}},
  pp. \bibinfo{pages}{268--276}, \doi{10.1007/978-3-540-32254-2\_16}.

\bibitemdeclare{incollection}{blackburn06}
\bibitem{blackburn06}
\bibinfo{author}{P.~\surnamestart Blackburn\surnameend} \&
  \bibinfo{author}{J.~\surnamestart van Benthem\surnameend}
  (\bibinfo{year}{2007}): \emph{\bibinfo{title}{Modal {L}ogic: {A} {S}emantic
  {P}erspective}}.
\newblock In: {\sl \bibinfo{booktitle}{Handbook of Modal Logic}},
  \bibinfo{publisher}{Elsevier}, pp. \bibinfo{pages}{1--84},
  \doi{10.1016/s1570-2464(07)80004-8}.

\bibitemdeclare{book}{blackburn01}
\bibitem{blackburn01}
\bibinfo{author}{P.~\surnamestart Blackburn\surnameend},
  \bibinfo{author}{M.~\surnamestart de~Rijke\surnameend} \&
  \bibinfo{author}{Y.~\surnamestart Venema\surnameend} (\bibinfo{year}{2001}):
  \emph{\bibinfo{title}{Modal Logic}}.
\newblock \bibinfo{series}{Cambridge Tracts in Theoretical Computer Science},
  \bibinfo{publisher}{Cambridge University Press},
  \doi{10.1017/CBO9781107050884}.

\bibitemdeclare{article}{blackburn95}
\bibitem{blackburn95}
\bibinfo{author}{P.~\surnamestart Blackburn\surnameend} \&
  \bibinfo{author}{J.~\surnamestart Seligman\surnameend}
  (\bibinfo{year}{1995}): \emph{\bibinfo{title}{Hybrid {L}anguages}}.
\newblock {\sl \bibinfo{journal}{Journal of Logic, Language and Information}}
  \bibinfo{volume}{4}(\bibinfo{number}{3}), pp. \bibinfo{pages}{251--272},
  \doi{10.1007/BF01049415}.

\bibitemdeclare{phdthesis}{tencate_phd}
\bibitem{tencate_phd}
\bibinfo{author}{B.~\surnamestart ten Cate\surnameend} (\bibinfo{year}{2005}):
  \emph{\bibinfo{title}{Model theory for extended modal languages}}.
\newblock Ph.D. thesis, \bibinfo{school}{University of Amsterdam}.
\newblock \bibinfo{note}{ILLC Dissertation Series DS-2005-01}.

\bibitemdeclare{inproceedings}{tencate2005complexity}
\bibitem{tencate2005complexity}
\bibinfo{author}{B.~\surnamestart ten Cate\surnameend} \&
  \bibinfo{author}{M.~\surnamestart Franceschet\surnameend}
  (\bibinfo{year}{2005}): \emph{\bibinfo{title}{On the complexity of hybrid
  logics with binders}}.
\newblock {\sl \bibinfo{series}{Lecture Notes in Computer Science}}
  \bibinfo{volume}{3634}, \bibinfo{publisher}{Springer Verlag}, pp.
  \bibinfo{pages}{339--354}, \doi{10.1007/11538363\_24}.

\bibitemdeclare{book}{vanditmarsch07}
\bibitem{vanditmarsch07}
\bibinfo{author}{H.~\surnamestart van Ditmarsch\surnameend},
  \bibinfo{author}{W.~\surnamestart van~der Hoek\surnameend} \&
  \bibinfo{author}{B.~\surnamestart Kooi\surnameend} (\bibinfo{year}{2007}):
  \emph{\bibinfo{title}{Dynamic Epistemic Logic}}.
\newblock \bibinfo{series}{Synthese Library}, \bibinfo{publisher}{Springer},
  \doi{10.1007/978-1-4020-5839-4}.

\bibitemdeclare{phdthesis}{fervari14phd}
\bibitem{fervari14phd}
\bibinfo{author}{R.~\surnamestart Fervari\surnameend} (\bibinfo{year}{2014}):
  \emph{\bibinfo{title}{Relation-Changing Modal Logics}}.
\newblock Ph.D. thesis, \bibinfo{school}{Universidad Nacional de C\'ordoba,
  Argentina}.

\bibitemdeclare{inproceedings}{FranceschetRS03}
\bibitem{FranceschetRS03}
\bibinfo{author}{M.~\surnamestart Franceschet\surnameend},
  \bibinfo{author}{M.~\surnamestart de~Rijke\surnameend} \&
  \bibinfo{author}{B.~\surnamestart Schlingloff\surnameend}
  (\bibinfo{year}{2003}): \emph{\bibinfo{title}{Hybrid Logics on Linear
  Structures: Expressivity and Complexity}}.
\newblock In: {\sl \bibinfo{booktitle}{{TIME-ICTL} 2003, Cairns, Queensland,
  Australia}}, pp. \bibinfo{pages}{166--173}, \doi{10.1109/time.2003.1214893}.

\bibitemdeclare{inproceedings}{GKVQ09}
\bibitem{GKVQ09}
\bibinfo{author}{N.~\surnamestart Gierasimczuk\surnameend},
  \bibinfo{author}{L.~\surnamestart Kurzen\surnameend} \&
  \bibinfo{author}{F.~R. \surnamestart Vel{\'a}zquez-Quesada\surnameend}
  (\bibinfo{year}{2009}): \emph{\bibinfo{title}{Learning and Teaching as a
  Game: A Sabotage Approach}}.
\newblock In \bibinfo{editor}{Xiangdong \surnamestart He\surnameend},
  \bibinfo{editor}{John~F. \surnamestart Horty\surnameend} \&
  \bibinfo{editor}{Eric \surnamestart Pacuit\surnameend}, editors: {\sl
  \bibinfo{booktitle}{LORI}}, {\sl \bibinfo{series}{Lecture Notes in Computer
  Science}} \bibinfo{volume}{5834}, \bibinfo{publisher}{Springer}, pp.
  \bibinfo{pages}{119--132}, \doi{10.1007/978-3-642-04893-7\_10}.

\bibitemdeclare{article}{gorausin92}
\bibitem{gorausin92}
\bibinfo{author}{V.~\surnamestart Goranko\surnameend} \&
  \bibinfo{author}{S.~\surnamestart Passy\surnameend} (\bibinfo{year}{1992}):
  \emph{\bibinfo{title}{Using the Universal Modality: Gains and Questions}}.
\newblock {\sl \bibinfo{journal}{Journal of Logic and Computation}}
  \bibinfo{volume}{2}(\bibinfo{number}{1}), pp. \bibinfo{pages}{5--30},
  \doi{10.1093/logcom/2.1.5}.

\bibitemdeclare{article}{Hoffmann2007}
\bibitem{Hoffmann2007}
\bibinfo{author}{G.~\surnamestart Hoffmann\surnameend} \&
  \bibinfo{author}{C.~\surnamestart Areces\surnameend} (\bibinfo{year}{2009}):
  \emph{\bibinfo{title}{HTab: A Terminating Tableaux System for Hybrid Logic}}.
\newblock {\sl \bibinfo{journal}{Electronic Notes in Theoretical Computer
  Science}} \bibinfo{volume}{231}, pp. \bibinfo{pages}{3--19},
  \doi{10.1016/j.entcs.2009.02.026}.

\bibitemdeclare{inproceedings}{loding03fsttcs}
\bibitem{loding03fsttcs}
\bibinfo{author}{C.~\surnamestart L\"oding\surnameend} \&
  \bibinfo{author}{P.~\surnamestart Rohde\surnameend} (\bibinfo{year}{2003}):
  \emph{\bibinfo{title}{Model {C}hecking and {S}atisfiability for {S}abotage
  {M}odal {L}ogic}}.
\newblock {\sl \bibinfo{series}{Lecture Notes in Computer Science}}
  \bibinfo{volume}{2914}, pp. \bibinfo{pages}{302--313},
  \doi{10.1007/978-3-540-24597-1\_26}.

\bibitemdeclare{mastersthesis}{martel15}
\bibitem{martel15}
\bibinfo{author}{M.~\surnamestart Martel\surnameend} (\bibinfo{year}{2015}):
  \emph{\bibinfo{title}{On the {U}ndecidability of {R}elation-{C}hanging
  {L}ogics}}.
\newblock Master's thesis, \bibinfo{school}{Universidad Nacional de R\'io
  Cuarto, Argentina}.

\bibitemdeclare{phdthesis}{rohde06phd}
\bibitem{rohde06phd}
\bibinfo{author}{P.~\surnamestart Rohde\surnameend} (\bibinfo{year}{2006}):
  \emph{\bibinfo{title}{On games and logics over dynamically changing
  structures}}.
\newblock Ph.D. thesis, \bibinfo{school}{RWTH Aachen}.

\bibitemdeclare{phdthesis}{schneider07phd}
\bibitem{schneider07phd}
\bibinfo{author}{T.~\surnamestart Schneider\surnameend} (\bibinfo{year}{2007}):
  \emph{\bibinfo{title}{The Complexity of Hybrid Logics over Restricted Frame
  Classes}}.
\newblock Ph.D. thesis, \bibinfo{school}{University of Jena}.

\bibitemdeclare{phdthesis}{Spaan93}
\bibitem{Spaan93}
\bibinfo{author}{E.~\surnamestart Spaan\surnameend} (\bibinfo{year}{1993}):
  \emph{\bibinfo{title}{Complexity of modal logics}}.
\newblock Ph.D. thesis, \bibinfo{school}{ILLC, University of Amsterdam}.

\end{thebibliography}
